\documentclass[11pt,a4paper]{article}

\usepackage[utf8]{inputenc}
\usepackage[T1]{fontenc}

\usepackage{amsmath,amssymb,amsthm,mathtools}
\usepackage{algorithm}
\usepackage{algpseudocode}
\usepackage{booktabs}
\usepackage{graphicx}
\usepackage{subcaption}
\usepackage{xcolor}
\usepackage{hyperref}
\usepackage{cleveref}
\usepackage{tikz}
\usepackage[a4paper,margin=1in]{geometry}

\usetikzlibrary{calc,positioning,arrows.meta}

\theoremstyle{plain}
\newtheorem{theorem}{Theorem}[section]
\newtheorem{lemma}[theorem]{Lemma}

\newtheorem{corollary}[theorem]{Corollary}

\theoremstyle{definition}
\newtheorem{definition}[theorem]{Definition}
\newtheorem{remark}[theorem]{Remark}

\newcommand{\tw}{\operatorname{tw}}
\newcommand{\pw}{\operatorname{pw}}
\newcommand{\TC}{\operatorname{TC}}

\newcommand{\OPT}{\operatorname{OPT}}
\newcommand{\Prob}{\mathbb{P}}

\title{
Exact Greedy Influence Maximization in Linear Time \\ on Bounded-Treewidth Graphs
}

\author{
Matic Požar\\
UP FAMNIT, University of Primorska\\
Koper, Slovenia\\
\texttt{matic.pozar@upr.si}
}

\date{}

\begin{document}

\maketitle


\begin{abstract}
Computing influence spread under the Independent Cascade (IC) model is
\#P-hard in general, and influence maximization is therefore commonly
approached using Monte Carlo simulation or reverse-reachable-set
sampling.

We study IC diffusion on graphs of bounded treewidth. Prior work on
probabilistic reachability and provenance establishes the tractability
of exact probabilistic reachability on bounded-treewidth instances.
We give an explicit IC formulation based on probability distributions
over separator reachability relations, yielding exact influence
evaluation in

$$
    O\!\left(
        n\,2^{O(w^2)}\operatorname{poly}(w)
    \right)
$$

time, for a graph with \(n\) nodes and treewidth \(w\).

Our main contribution is an exact all-marginal-gains algorithm.
We introduce variable artificial source edges whose probabilities
parameterize potential seed selection and show that, at a deterministic
seed set, the partial derivative with respect to each source-edge
probability is exactly the corresponding greedy marginal gain.
Reverse-mode differentiation through the exact inference computation
therefore obtains all node marginal gains simultaneously, with the same
asymptotic complexity as one exact influence evaluation.

This yields an exact implementation of the classical greedy influence
maximization algorithm in

$$
    O\!\left(
        K n\,2^{O(w^2)}\operatorname{poly}(w)
    \right)
$$

time, which is linear in graph size for fixed treewidth $w$ and seed
budget $K$. We further show that the separator-relation representation
has tight
\( 2^{\Theta(w^2)}\)
state complexity within the class of exact context-independent
compositional separator summaries. This tractability of exact greedy
optimization contrasts with the NP-hardness of globally optimal IC
influence maximization already on graphs of treewidth one and
pathwidth two.

Experiments on synthetic bounded-treewidth networks are consistent
with the predicted linear scaling in graph size for fixed width and
show that the runtime is largely insensitive to the numerical
propagation and seed-activation probabilities. In demanding diffusion
regimes, the proposed method substantially outperforms
reverse-reachable-set and optimized Monte Carlo greedy baselines while
computing the classical greedy marginal gains exactly.
\end{abstract}


\section{Introduction}
\label{sec:introduction}

The Influence Maximization (IM) problem~\cite{kempe_im} asks for a set
of at most $K$ seed vertices whose initial activation maximizes the
expected extent of a subsequent diffusion process. One of the most
widely studied diffusion models is the Independent Cascade (IC) model.
Influence maximization under IC is NP-hard, while exactly computing the
expected influence spread of a fixed seed set is itself
\#P-hard~\cite{mia_icphard}.

Consequently, practical IM algorithms generally rely on approximate
influence evaluation or randomized sampling. A classical approach is
Monte Carlo (MC) simulation of the diffusion process. However, greedy
influence maximization requires repeated influence evaluations for many
candidate vertices, making straightforward MC-based optimization
expensive even when the number of evaluations is reduced through lazy
submodular techniques such as CELF~\cite{CELF}.

A major improvement is obtained through reverse-reachable (RR) set
sampling. Algorithms such as TIM/TIM+~\cite{tim} and IMM~\cite{IMM}
provide strong approximation guarantees together with near-linear
theoretical running times. Their practical computational requirements,
however, remain sampling-dependent and can vary substantially across
diffusion regimes and target accuracy levels.

An alternative source of tractability is structural restriction of the
underlying graph. Many graph problems that are intractable in general
become efficiently solvable when the treewidth \cite{cygan2015parameterized} of the input graph is
bounded~\cite{arnborg1991easy,bodlaender1996linear}. We denote the
treewidth by $w$.

Relatively little work has examined the consequences of bounded
treewidth for IC influence maximization. Influence maximization on
trees admits an FPTAS~\cite{FPTAS-tree}; the same work notes that the
approach can be extended to bounded-treewidth graphs, although at a
substantial increase in complexity, without developing this extension
further. More recently, exact IC influence spread has been shown to be
computable in linear time on graphs of bounded pathwidth
~\cite{bounded_pathwidth}.

Closely related results also arise in probabilistic databases and
provenance. Under the live-edge interpretation of IC, every directed
edge can be viewed as an independently present probabilistic tuple, and
the activation probability of a vertex becomes a probabilistic
reachability query. General bounded-treewidth probabilistic-query
machinery therefore provides exact structural inference for such
queries, while recent provenance systems additionally implement shared
all-target reachability compilation on bounded-treewidth
instances~\cite{provsql}. These results suggest that exact IC influence
evaluation on bounded-treewidth graphs should be viewed as an instance
of a broader class of tractable probabilistic-reachability problems.

This paper asks a different question: \emph{what does this structural
tractability of exact influence evaluation imply for influence
maximization itself?} Our main observation is that exact evaluation of
a single influence function can be extended to provide all greedy
marginal gains simultaneously. By introducing variable artificial
source edges and differentiating the resulting exact inference
computation, the marginal gain of every possible new seed is recovered
from a single reverse pass.

The main contributions are as follows:

\begin{enumerate}
\item
We give an explicit exact IC inference formulation for
bounded-treewidth graphs using probability distributions over
transitively closed separator reachability relations.

\item
We show that exact influence evaluation and all vertex activation
marginals can be computed in
\[
    O\!\left(
        n\,2^{O(w^2)}\operatorname{poly}(w)
    \right)
\]
arithmetic operations.

\item
We introduce source-edge variables $\theta_v$ and prove that, for
every $v\notin S$,
\[
    \frac{\partial F}{\partial \theta_v}(\mathbf 1_S)
    =
    \sigma(S\cup\{v\})-\sigma(S),
\]
where \(F\) is a differentiable function that computes the exact IC
diffusion in bounded treewidth graphs.
Reverse-mode differentiation therefore computes all exact greedy
marginal gains simultaneously within the same asymptotic
complexity as one exact influence evaluation.

\item
We obtain an exact implementation of classical greedy IM with
running time
\[
    O\!\left(
        K n\,2^{O(w^2)}\operatorname{poly}(w)
    \right).
\]
Hence, for fixed treewidth $w$ and seed budget $K$, exact greedy IM
is linear in the graph size.

\item
We show that the separator-relation representation has \(2^{\Theta(w^2)}\)
worst-case state complexity within the class of exact
context-independent compositional separator summaries.

\item
We experimentally compare the proposed method with IMM and a
directed NewGreedy/Cohen-style all-gains Monte Carlo baseline,
studying both scaling in graph size and sensitivity to stochastic
diffusion parameters.
\end{enumerate}


\section{Preliminaries and Related Work}
\label{sec:preliminaries}

\subsection{Independent Cascade Model}

Let
\(
    G=(V,E)
\)
be a simple directed graph with \(n=|V|\).
Each edge $e=(u,v)\in E$ is associated with an independent propagation
probability
\(
    p_e\in[0,1],
\)
which represents the probability with which newly active nodes activate their inactive neighbors.

Given a seed set $S\subseteq V$, the Independent Cascade process starts
with all vertices of $S$ active. Then the diffusion proceeds in discrete steps.
Whenever \(u\) becomes active at time \(t\), it gets one independent opportunity to activate each currently inactive out-neighbor \(v\) with probability \(p_{uv}\). Successfully activated neighbors then attempt to activate their own inactive out-neighbors at time \(t+1\). The process terminates when no new activations occur.

What is of interest is not any particular unfolding of the IC, but rather the probability with which each node is activated.
Let
\(
    \sigma(S)
\)
denote the expected number of activated vertices, which equals the sum of activation probabilities of all the nodes.

\subsection{Live-Edge Representation}

The IC model, and more generally triggering models \cite{kempe_im,kempe_UshareIshare}, admits an equivalent live-edge formulation.
Since all the influence attempts are independent, each edge \(e\) can be regarded as live with probability \(p_e\) and dead with probability \(1-p_e\).
A vertex $v$ is activated by seed set $S$ iff it is reachable from
$S$ in the resulting live-edge graph.

Thus
\[
    \Prob(v\text{ active}\mid S)
    =
    \Prob(S\rightsquigarrow v).
\]

This equivalence turns IC inference into probabilistic reachability.

\subsection{Influence Maximization}

The influence maximization problem asks for
\[
    S^\star
    \in
    \arg\max_{S\subseteq V,\ |S|\le K}
    \sigma(S).
\]

For IC, $\sigma$ is monotone and submodular.
Therefore the classical greedy algorithm, which repeatedly selects the
node with maximum marginal gain,
\(
    \Delta(v\mid S)
    =
    \sigma(S\cup\{v\})-\sigma(S),
\)
achieves the approximation guarantee
\[
    \sigma(S_{\mathrm{greedy}})
    \ge
    \left(1-\frac1e\right)
    \sigma(S^\star).
\]

The difficulty is that exact marginal evaluation is itself
computationally expensive in general graphs.

\subsection{Tree Decompositions}

\begin{definition}[Tree decomposition]
A tree decomposition of an undirected graph $G$ is a pair
\[
    (\mathcal{T},\{B_t\}_{t\in V(\mathcal T)})
\]
where each $B_t\subseteq V(G)$ is a bag satisfying:
\begin{enumerate}
    \item every vertex occurs in at least one bag;
    \item every edge has both endpoints in at least one bag;
    \item the bags containing any fixed vertex form a connected
    subtree of $\mathcal T$.
\end{enumerate}
\end{definition}

The width of a decomposition is
\[
    \max_t |B_t|-1,
\]
and the treewidth $\tw(G)$ is the minimum width over all tree
decompositions.

For directed graphs we use the treewidth of the underlying undirected
graph.

\subsection{Relation to Probabilistic Databases}

The live-edge interpretation of the Independent Cascade model is closely related to tuple-independent probabilistic databases. In a tuple-independent database, every input fact is present independently with an associated probability. Hence, representing every directed edge $(u,v)$ as an uncertain tuple whose presence probability is $p_{uv}$ produces exactly the same distribution over possible graphs as the live-edge formulation of the Independent Cascade model. Under this correspondence, the activation probability of a vertex $v$ from a seed set $S$ is the probability of the recursive reachability query $S\rightsquigarrow v$.

Probabilistic query evaluation is generally \#P-hard, but structural restrictions on the underlying data can make it tractable. Amarilli, Bourhis, and Senellart \cite{Provenance} showed that, for fixed logical queries on instances of bounded treewidth, provenance circuits can be constructed in time linear in the size of the instance, yielding tractable exact probabilistic query evaluation on treelike instances. In particular, directed probabilistic reachability can therefore be computed exactly in linear time for fixed treewidth.

This connection is implemented explicitly in the recent ProvSQL \cite{provsql} system, whose bounded-treewidth reachability compiler constructs reachability provenance along a tree decomposition of the probabilistic graph \footnote{\href{https://provsql.org/docs/user/probabilities.html}{provSQL Probabilities documentation page.}}. The current implementation additionally supports an all-target compilation in which the reachability events of all vertices share a single linear-size circuit for fixed treewidth.

In the influence-maximization literature, Nakamura and Nishino \cite{bounded_pathwidth} independently considered the more restrictive bounded-pathwidth setting and showed that all activation probabilities, and hence the complete IC influence spread, can be computed exactly in
\(
O\!\left((m+n)\omega_p^2 2^{\omega_p^2}\right)
\)
time. Our construction can be viewed as a tree-decomposition formulation of the same underlying structural principle: a subgraph can be summarized, with respect to the rest of the graph, by its reachability behavior on a bounded-size separator. Our subsequent contribution is not the tractability of bounded-treewidth probabilistic reachability itself, but the use of this exact inference structure to compute all greedy influence-maximization marginal gains simultaneously.

We use this existing structural inference principle as the foundation
for the influence-maximization results developed below.


\section{Exact IC Inference on Bounded-Treewidth Graphs}
\label{sec:exact-ic}

In this section, the graphs in question are considered to be simple.

\subsection{Separator Reachability Relations}

Consider an oriented tree-decomposition edge
\(
    x\to y
\)
with separator
\(
    C_{xy}=B_x\cap B_y.
\)

For a fixed live-edge realization on the $x$-side of the decomposition
edge, define a relation
\(
    R_{x\to y}
    \subseteq
    C_{xy}\times C_{xy}
\)
by
\(
    (u,v)\in R_{x\to y}
    \iff
    u\rightsquigarrow v
\)
using only vertices and edges on the $x$-side.

\begin{lemma}[Separator sufficiency]
\label{lem:separator-sufficiency}
Let $H$ be a directed graph component whose only intersection with the
remainder of the graph is a separator $C$.
For a fixed realization of $H$, let
\[
    R_H
    =
    \{(u,v)\in C\times C:
      u\rightsquigarrow_H v\}.
\]
Replacing $H$ by directed arcs representing $R_H$ preserves all
reachability relations among vertices outside $H$ and vertices of $C$.
\end{lemma}

\begin{proof}
Any path entering $H$ can be decomposed into maximal excursions from a
separator vertex $u\in C$ to a separator vertex $v\in C$.
Each such excursion can be replaced by the relation arc $(u,v)$.

Conversely, replacing \(H\) by \(R_H\) introduces no new reachability. Every relation arc \((u,v)\in R_H\) exists only because there is a directed path from \(u\) to \(v\) inside \(H\). Hence any path in the compressed graph can be converted into a path in the original graph by replacing each relation arc by its witnessing path in \(H\).
\end{proof}

\subsection{Probabilistic Messages}
\label{subsec:prob-messages}

Because the IC live edges are random, a directed decomposition message
is a probability distribution over separator reachability relations.
For a directed decomposition edge $x\to y$, let
\(C_{xy}=B_x\cap B_y\)
be the corresponding separator. We define
$$
    M_{x\to y}(R)
    =
    \Prob(R_{x\to y}=R),
$$
where $R_{x\to y}$ is the reachability relation induced on $C_{xy}$
by the part of the graph represented on the $x$-side of the
decomposition edge ${x,y}$.

For a bag $x$, let
\(P_x^{\mathrm{loc}}(R)\)
denote the probability that the uncertain edges owned locally by $x$
induce the reachability relation
\(R\subseteq B_x\times B_x.\)

\paragraph{Unique ownership of uncertain edges.}
Every uncertain graph edge is assigned to exactly one bag containing
both of its endpoints, and the corresponding live-edge variable is
sampled only at this owning bag. This convention is necessary because
the same graph edge may occur in several bags of a tree decomposition.
Sampling it independently at every such occurrence would incorrectly
replace one IC activation attempt by several independent random events.

Unique ownership partitions the uncertain edge variables among the
bags. Consequently, after removing a decomposition node $x$, the edge
variables owned in distinct incident branches are disjoint from one
another and from those owned locally by $x$. Since IC live-edge
variables are independent, the random relation produced by each branch
is therefore independent of the local relation and of the relations
produced by the other branches before they are combined at $x$.

This independence is what allows the probability of a fixed tuple of
relation states to factor as a product.

\paragraph{Message recursion.}
Consider a directed decomposition edge $x\to y$, and let

$$
    z_1,\ldots,z_d
    =
    N_{\mathcal T}(x)\setminus\{y\}
$$
be the neighbors of \(x\) other than \(y\).

Suppose the local state is $R_x^{\mathrm{loc}}$, and the incoming
branches induce states
\(R_1,\ldots,R_d.\)
For this fixed tuple, define the outgoing separator relation by their transitive closure restricted to \(C_{xy}\times C_{xy}\) as
$$
    \Psi_{x\to y}
    (R_x^{\mathrm{loc}},R_1,\ldots,R_d)
    =
    \left.
    \TC\left(
        R_x^{\mathrm{loc}}
        \cup
        R_1
        \cup\cdots\cup
        R_d
    \right)
    \right|_{C_{xy}\times C_{xy}}.
$$
Note that the relations are aggregated with the \(\TC\) operation, since the individual reachability information of each 
relation may imply new reachability states not present in their union.

The probability of this tuple is

$$
    P_x^{\mathrm{loc}}(R_x^{\mathrm{loc}})
    \prod_{i=1}^{d}
    M_{z_i\to x}(R_i),
$$

because the corresponding edge-variable sets are disjoint and
independent. Here \(P_x^{\mathrm{loc}}(R_x^{\mathrm{loc}})\) denotes the locally computed probability of
relation \(R_x^{\mathrm{loc}}\).

Different tuples may produce the same outgoing relation. Their
probability masses are therefore accumulated, giving the recursion

$$
\boxed{
M_{x\to y}(R)
=
\sum_{\substack{
R_x^{\mathrm{loc}},R_1,\ldots,R_d:\\
\Psi_{x\to y}
(R_x^{\mathrm{loc}},R_1,\ldots,R_d)=R
}}
P_x^{\mathrm{loc}}(R_x^{\mathrm{loc}})
\prod_{i=1}^{d}
M_{z_i\to x}(R_i).
}
$$

Thus the relation operation
$\Psi_{x\to y}$ determines which outgoing state receives each
contribution, while the product of the local and incoming message
probabilities determines the mass of that contribution.

\begin{lemma}[Exactness of the message recursion]
\label{lem:message-correctness}
Consider a directed decomposition edge $x\to y$.
Remove the decomposition edge ${x,y}$, and let
$\mathcal T_{x\to y}$ denote the component containing $x$.
Assume that every uncertain graph edge has exactly one owning bag and
is sampled only at that bag.

Then, for every relation
\(R\subseteq C_{xy}\times C_{xy},\)
the message recursion computes exactly the probability that the edges
whose owners lie in $\mathcal T_{x\to y}$ induce relation $R$ on
$C_{xy}$:

$$
    M_{x\to y}(R)
    =
    \Prob(R_{x\to y}=R).
$$

\end{lemma}

\begin{proof}
We prove the claim by induction on the number of bags in
$\mathcal T_{x\to y}$.

Let
\(
    z_1,\ldots,z_d
    =
    N_{\mathcal T}(x)\setminus\{y\}.
\)
By the induction hypothesis, each incoming message
\(M_{z_i\to x}\)
is the exact distribution of the separator relation induced by the
corresponding branch.

Fix one realization of the uncertain edges on the $x$-side, and let
$R_x^{\mathrm{loc}},R_1,\ldots,R_d$ be the resulting local and
incoming relation states. By
Lemma~\ref{lem:separator-sufficiency}, each incoming branch may be
replaced by its separator relation without changing reachability
outside that branch. Hence the relation induced on the outgoing
separator is exactly

$$
    \Psi_{x\to y}
    (R_x^{\mathrm{loc}},R_1,\ldots,R_d).
$$

By unique edge ownership, the local edge variables and those belonging
to the incoming branches are disjoint. Their relation states are
therefore independent, so the probability of the fixed tuple is

$$
    P_x^{\mathrm{loc}}(R_x^{\mathrm{loc}})
    \prod_{i=1}^{d}
    M_{z_i\to x}(R_i).
$$

The recursion sums these probabilities over exactly those tuples whose
image under $\Psi_{x\to y}$ is $R$. Since these tuples partition the
live-edge realizations inducing outgoing relation $R$, the resulting
sum is precisely

$$
    \Prob(R_{x\to y}=R).
$$

For the base case, $x$ has no neighbor other than $y$, so the outgoing
message is obtained directly from the local relation distribution by
transitive closure and restriction to $C_{xy}$.
\end{proof}

\subsection{Upward and Downward Message Passing}
\label{subsec:two-sweeps}

Lemma~\ref{lem:message-correctness} gives the exact recursion for a
directed decomposition message $M_{x\to y}$, assuming that all incoming
messages
\[
    M_{z\to x},
    \qquad
    z\in N_{\mathcal T}(x)\setminus\{y\},
\]
are already available.
To recover reachability information from the entire graph at every bag,
we therefore compute the message in both directions across every
decomposition edge.

Root the decomposition tree $\mathcal T$ at an arbitrary bag $r$.
The messages are evaluated in two passes.

\paragraph{Upward sweep.}
Process the bags from the leaves toward the root.
For every non-root bag $x$ with parent $y$, compute
\[
    M_{x\to y}
\]
using the recursion of Lemma~\ref{lem:message-correctness}.
All required incoming messages $M_{z\to x}$ come from children of $x$
and have therefore already been computed.
For a leaf, there are no incoming branch messages and the message is
obtained directly from the local relation distribution.

\paragraph{Downward sweep.}
Process the bags from the root toward the leaves.
For every edge in which $x$ is the parent of $y$, compute
\[
    M_{x\to y}
\]
using the same recursion.
The messages from the other children of $x$ were computed during the
upward sweep, while, if $x\neq r$, the message from the parent of $x$
was computed earlier in the downward sweep.
Thus every incoming message required for $M_{x\to y}$ is available
when the edge is processed.

The two sweeps therefore use the same symmetric update
\[
    M_{x\to y}
    =
    \Phi_{x\to y}
    \left(
        P_x^{\mathrm{loc}},
        \{M_{z\to x}:
        z\in N_{\mathcal T}(x)\setminus\{y\}\}
    \right),
\]
where $\Phi_{x\to y}$ denotes the relation-combination, transitive
closure, restriction, and probability accumulation operation defined
in Section~\ref{subsec:prob-messages}.
The distinction between upward and downward messages is only their
evaluation order.

\begin{lemma}[Exactness of the two-sweep computation]
\label{lem:two-sweeps}
After one upward sweep followed by one downward sweep, both directed
messages associated with every decomposition edge are exact.
\end{lemma}

\begin{proof}
During the upward sweep, every child-to-parent message is evaluated
only after all messages entering the child from its other neighbors
have been computed. Hence Lemma~\ref{lem:message-correctness} implies
that every upward message is exact.

During the downward sweep, when computing a parent-to-child message
$M_{x\to y}$, every required message entering $x$ from a neighbor
other than $y$ is already exact: child-to-$x$ messages were computed
during the upward sweep, and the parent-to-$x$ message, when one
exists, was computed earlier in the downward sweep.
Applying Lemma~\ref{lem:message-correctness} again therefore makes
every downward message exact.
Thus both orientations of every decomposition edge are computed
exactly.
\end{proof}

\subsubsection{Adding the seed set as an artificial source}

We next convert reachability from a seed set into ordinary
single-source reachability.

Introduce a new vertex
\(\rho\)
that does not belong to the original graph.
For every seed vertex
\(s\in S,\)
add a deterministic directed edge
\(\rho\to s.\)
These artificial edges have probability $1$.
No edge $\rho\to v$ is added for a non-seed vertex $v$.

To obtain a valid tree decomposition of the augmented graph, add
$\rho$ to every bag:
\(
B_x^+
    =
    B_x\cup\{\rho\}.
\)
Because $\rho$ now occurs in every bag, the connectedness condition of
a tree decomposition is automatically satisfied.
Moreover, every artificial edge $\rho\to s$ is contained in every bag
that originally contained $s$, so it can be assigned to one such bag
according to the same unique-ownership convention used for ordinary
edges.

If the original decomposition has width $w$, the augmented
decomposition has width at most \(w+1.\)
Thus the artificial source changes the width only by one.

The purpose of $\rho$ is that, in every fixed live-edge realization,
\(\rho\rightsquigarrow v\)
holds if and only if there exists a seed
\(s\in S\)
such that
\( s\rightsquigarrow v\)
through live original edges.
Indeed, any path starting at $\rho$ must first use one of the
deterministic edges $\rho\to s$, and every seed $s$ is reachable
directly from $\rho$.
Therefore
$$
    \Prob(v\text{ is activated by }S)
    =
    \Prob(\rho\rightsquigarrow v).
$$
The artificial source is retained in all separator and bag relations
throughout the dynamic program.
In particular, after augmentation the separator associated with
${x,y}$ is
\(
C_{xy}^+
    =
    (B_x\cap B_y)\cup\{\rho\}.
\)

\begin{lemma}[Artificial-source equivalence]
\label{lem:source-equivalence}
For every live-edge realization of the original IC graph and every
vertex $v\in V$,

$$
    v\text{ is activated from }S
    \quad\Longleftrightarrow\quad
    \rho\rightsquigarrow v
$$

in the augmented graph.
Consequently,

$$
    \Prob(v\text{ active}\mid S)
    =
    \Prob(\rho\rightsquigarrow v).
$$

\end{lemma}

\begin{proof}
If $v$ is activated from $S$, then some seed $s\in S$ has a live
directed path to $v$.
Since $\rho\to s$ is deterministic, prefixing this path with
$\rho\to s$ yields a path from $\rho$ to $v$.

Conversely, every path from $\rho$ to an original vertex must leave
$\rho$ through an artificial edge $\rho\to s$ for some $s\in S$.
The remainder of the path consists of live original edges, so $v$ is
reachable from a seed and is therefore activated under the IC
live-edge interpretation.
\end{proof}

\subsubsection{Global Bag Distributions and Activation Marginals}
\label{subsec:Global Bag Distributions and Activation Marginals}

After the upward and downward sweeps, every bag $x$ has an exact
incoming message
\(M_{z\to x}\)
from each neighboring decomposition branch
$z\in N_{\mathcal T}(x)$.
Together with the locally owned edges at $x$, these branches account
for every uncertain edge of the graph exactly once.

For a local relation $R_x^{\mathrm{loc}}$ and incoming branch
relations $\{R_z\}_{z\in N_{\mathcal T}(x)}$, define the resulting
relation on the augmented bag $B_x^+$ by

$$
    \Psi_x
    \left(
        R_x^{\mathrm{loc}},
        \{R_z\}
    \right)
    =
    \left.
    \TC\left(
        R_x^{\mathrm{loc}}
        \cup
        \bigcup_{z\in N_{\mathcal T}(x)}R_z
    \right)
    \right|_{B_x^+\times B_x^+}.
$$

The corresponding global bag distribution is

$$
\boxed{
P_x(R)
=
\sum_{\substack{
R_x^{\mathrm{loc}},\{R_z\}:\\
\Psi_x(R_x^{\mathrm{loc}},\{R_z\})=R
}}
P_x^{\mathrm{loc}}(R_x^{\mathrm{loc}})
\prod_{z\in N_{\mathcal T}(x)}
M_{z\to x}(R_z).
}
$$

By Lemma~\ref{lem:two-sweeps}, all incoming messages are exact, and by
Lemma~\ref{lem:separator-sufficiency}, replacing each incident branch
by its separator relation preserves reachability on $B_x^+$.
Therefore $P_x(R)$ is exactly the probability that the complete
augmented live-edge graph induces relation $R$ on $B_x^+$.

Now let $v$ be an original graph vertex and let $x$ be any bag
containing $v$. Since $\rho$ belongs to every augmented bag,

$$
    \Prob(\rho\rightsquigarrow v)
    =
    \sum_R
    P_x(R)\,
    \mathbf 1[(\rho,v)\in R].
$$

By Lemma~\ref{lem:source-equivalence},

$$
\boxed{
    \pi_v(S)
    =
    \Prob(v\text{ active}\mid S)
    =
    \sum_R
    P_x(R)\,
    \mathbf 1[(\rho,v)\in R].
}
$$

A vertex may occur in several bags, but the value above is independent
of the chosen bag because every $P_x$ is a restriction of the same
global live-edge distribution. We therefore assign each vertex $v$ an
arbitrary home bag $h(v)$ containing it and evaluate $\pi_v(S)$ only
there.

Finally,

$$
    \sigma(S)
    =
    \sum_{v\in V}\pi_v(S).
$$

\subsection{Running Time}

A bag contains at most
\( w+1\)
graph vertices.
Including the artificial source gives $O(w)$ boundary objects (nodes of the original graph).
A directed reachability relation can therefore be represented by
$O(w^2)$ Boolean entries.
Hence there are at most
\(2^{O(w^2)}\)
possible relation states. Here we measure complexity as the number of arithmetic operations.

\begin{theorem}[Exact bounded-treewidth IC inference]
\label{thm:exact-ic}
Given a tree decomposition of width $w$ with \(O(n)\) bags, all IC activation marginals
can be computed exactly in

$$
    O\!\left(
        n\,2^{O(w^2)}\operatorname{poly}(w)
    \right)
$$

arithmetic operations.
For fixed $w$, the running time is linear in $n$.
\end{theorem}

\begin{proof}
By Lemmas~\ref{lem:message-correctness} and~\ref{lem:two-sweeps}, one
upward and one downward sweep compute the exact directed message on
every decomposition edge. A width-$w$ bag, after inclusion of the
artificial source, contains $O(w)$ boundary objects and therefore
admits at most
\(2^{O(w^2)}\)
directed reachability states. Combining two relation distributions,
including the required transitive closure and restriction, therefore
takes
\(2^{O(w^2)}\operatorname{poly}(w)\)
arithmetic operations.

Each directed decomposition edge is processed once in each sweep.
Likewise, the global bag distribution $P_x$ from Section~\ref{subsec:Global Bag Distributions and Activation Marginals} is obtained by successively
combining the local relation distribution at $x$ with the incoming
messages
\(\{M_{z\to x}:z\in N_{\mathcal T}(x)\},\)
using the same relation-combination operation as in the message
recursion. The work needed to construct \(P_x^{\text{loc}}\) for each bag \(x\) is \(O(2^{O(w^2)}\operatorname{poly}(w))\) since the size of each bag is \(O(w)\) with \(O(w^2)\) edges.
Hence the work at bag $x$ is proportional to
$\deg_{\mathcal T}(x)$ times
$2^{O(w^2)}\operatorname{poly}(w)$, and since

$$
    \sum_x \deg_{\mathcal T}(x)=O(n)
$$

for a linear-size tree decomposition, constructing all global bag
distributions requires

$$
    O\!\left(
        n\,2^{O(w^2)}\operatorname{poly}(w)
    \right)
$$

operations.

Finally, each graph vertex is evaluated only at its home bag. Computing

$$
    \pi_v(S)
    =
    \sum_R
    P_{h(v)}(R)\,
    \mathbf 1[(\rho,v)\in R]
$$

requires scanning at most $2^{O(w^2)}$ relation states, so extracting
all $n$ activation marginals also costs

$$
    O\!\left(
        n\,2^{O(w^2)}
    \right).
$$

Combining these bounds gives the claimed

$$
    O\!\left(
        n\,2^{O(w^2)}\operatorname{poly}(w)
    \right)
$$

running time.
\end{proof}

\subsection{State-Complexity Lower Bound}

We next show that the quadratic exponent in the separator-relation
representation is unavoidable within the class of
context-independent compositional summaries.

Let $C$ be a separator of size $\Theta(w)$ and partition it into two
sets

$$
    C=L\cup R,
    \qquad
    |L|,|R|=\Theta(w).
$$

We construct a family of deterministic subgraphs whose only interface
with the remainder of the graph is the separator $C$.
These subgraphs should be thought of as possible components lying on
one side of a decomposition edge.

For every subset
\(X\subseteq L\times R,\)
define one such subgraph $H_X$ as follows.
For every pair
\((\ell,r)\in X,\)
include the directed edge
\(\ell\to r\)
with activation probability $1$, and include no edge from $\ell$ to
$r$ when
\((\ell,r)\notin X.\)
No other edges are present in $H_X$.

Thus $H_X$ is deterministic: there is only one live-edge realization.
Moreover, because every edge is directed from $L$ to $R$ and there are
no edges leaving vertices of $R$, no directed path can use more than
one edge. Consequently, the reachability relation induced by $H_X$ on
the separator is exactly

$$
    \ell\rightsquigarrow r
    \quad\Longleftrightarrow\quad
    (\ell,r)\in X.
$$

Hence every subset
\(X\subseteq L\times R\)
corresponds to a distinct possible separator behavior.
Since there are
\(2^{|L||R|}
    =
    2^{\Theta(w^2)}\)
such subsets, there are
\(2^{\Theta(w^2)}\)
distinct reachability behaviors that a separator summary may have to
represent.

\begin{theorem}[Separator-state lower bound]
\label{thm:state-lower-bound}
Any exact context-independent compositional summary of a separator of
size $\Theta(w)$ that preserves reachability under arbitrary external
contexts requires

$$
    2^{\Omega(w^2)}
$$

distinguishable states in the worst case.
\end{theorem}

\begin{proof}
The construction above gives
\(
    2^{|L||R|}
    =
    2^{\Theta(w^2)}
\)
different separator reachability relations.
It remains to show that an exact context-independent summary cannot
merge two distinct such relations into the same state.

Consider two distinct subsets
\(X,Y\subseteq L\times R.\)
Since $X\neq Y$, there exists a pair
\((\ell,r)\in X\triangle Y.\)
Without loss of generality, suppose
\((\ell,r)\in X\) and \((\ell,r)\notin Y.\)

We now construct an external context that distinguishes the two branch
behaviors. Add a new source vertex $\rho$ with the deterministic edge
\(\rho\to\ell,\)
and add a new target vertex $t$ with the deterministic edge
\(r\to t.\)
No other edges are added by the external context.

When the branch $H_X$ is attached, the relation
\( \ell\rightsquigarrow r\)
holds, and therefore
\(\rho\rightsquigarrow t.\)
When $H_Y$ is attached, this reachability does not hold. Since the
construction of $H_Y$ contains only edges from $L$ to $R$, there is no
alternative path from $\ell$ to $r$, and hence
\(\rho\not\rightsquigarrow t.\)
Thus $H_X$ and $H_Y$ produce different reachability answers under the
same external context. Any exact context-independent compositional
summary must therefore assign them different states; otherwise the
remainder of the computation could not distinguish which answer is
correct.

Since this argument applies to every pair of distinct subsets
$X,Y\subseteq L\times R$, all
\(
2^{|L||R|}
    =
    2^{\Theta(w^2)}
\)
separator behaviors must be distinguishable. Hence at least
\(2^{\Omega(w^2)}\)
summary states are required in the worst case.
\end{proof}

\begin{remark}
Theorem \ref{thm:state-lower-bound} is a lower bound for exact
context-independent separator summaries.
It is not claimed to be a universal lower bound for every possible
bounded-treewidth algorithm.
\end{remark}


\section{Exact All-Node Marginal Gains in Linear Time}
\label{sec:marginal-gains}

This section contains the main influence-maximization observation.

\subsection{Variable Seed Edges}

We now replace the deterministic artificial source edges used in the
previous section by probabilistic source edges.

Introduce the artificial source vertex $\rho$ as before. For every
original graph vertex $v\in V$, add an artificial directed edge
\(\rho\to v\)
whose live probability is a variable
\( \theta_v\in[0,1].\)
Let
\(
    \boldsymbol{\theta}
    =
    (\theta_v)_{v\in V}.
\)

The augmented live-edge model therefore contains two types of
independent Bernoulli variables:
\begin{itemize}
\item every original graph edge $e\in E$ is live independently with
its fixed IC probability $p_e$; and
\item every artificial source edge $\rho\to v$ is live independently
with probability $\theta_v$.
\end{itemize}

For any event $A$ in this augmented random graph, write
\( \Prob_{\boldsymbol{\theta}}(A)\)
for its probability under these edge probabilities. In particular,
\(\Prob_{\boldsymbol{\theta}}(\rho\rightsquigarrow u)\)
denotes the probability that $u$ is reachable from $\rho$ when the
original edges are sampled according to their fixed probabilities
${p_e}_{e\in E}$ and the artificial source edges are sampled
according to $\boldsymbol{\theta}$.

Define

$$
    F(\boldsymbol{\theta})
    =
    \sum_{u\in V}
    \Prob_{\boldsymbol{\theta}}
    (\rho\rightsquigarrow u).
$$

For a deterministic seed set $S\subseteq V$, let
\(\boldsymbol{\theta}=\mathbf 1_S,\)
where
$$
    (\mathbf 1_S)_v
    =
    \begin{cases}
        1,&v\in S,\\
        0,&v\notin S.
    \end{cases}
$$
Under this assignment, the artificial edge $\rho\to v$ is
deterministically present exactly when $v\in S$ and deterministically
absent otherwise. Hence the augmented model reduces exactly to the
artificial-source construction for the deterministic seed set $S$.
By Lemma~\ref{lem:source-equivalence},

$$
    \Prob_{\mathbf 1_S}(\rho\rightsquigarrow u)
    =
    \Prob(u\text{ is activated from }S),
$$

and therefore

$$
    \boxed{
    F(\mathbf 1_S)=\sigma(S).
    }
$$

\subsection{Derivative Identity}

\begin{theorem}[Exact marginal gains from derivatives]
\label{thm:gradient-marginals}
Let $S\subseteq V$ and let $v\notin S$. Then

$$
    \boxed{
    \frac{\partial F}{\partial \theta_v}(\mathbf 1_S)
    =
    \sigma(S\cup\{v\})-\sigma(S)
    }.
$$

\end{theorem}

\begin{proof}
Fix every coordinate except $\theta_v$ according to the indicator
vector $\mathbf 1_S$, and allow only the $v$-th coordinate to vary.
For $t\in[0,1]$, define

$$
    \boldsymbol{\theta}^{(v)}(t)
    =
    \mathbf 1_S+t\mathbf e_v,
$$

where $\mathbf e_v$ is the $v$-th standard basis vector.
Since $v\notin S$, this means

$$
    \theta_u^{(v)}(t)
    =
    \begin{cases}
        1,&u\in S,\\
        t,&u=v,\\
        0,&u\notin S\cup\{v\}.
    \end{cases}
$$

Under this assignment, the artificial source edges corresponding to
vertices in $S$ are deterministically live, the source edges
corresponding to vertices outside $S\cup{v}$ are deterministically
absent, and the single remaining artificial edge
\(\rho\to v\)
is live with probability $t$.

Conditioning on the state of this edge gives
$$
    F(\boldsymbol{\theta}^{(v)}(t))
    =
    (1-t)\sigma(S)
    +
    t\,\sigma(S\cup\{v\}).
$$

Indeed, when $\rho\to v$ is absent, the deterministic seed set is
exactly $S$, whereas when it is present, the deterministic seed set is
$S\cup{v}$.

Differentiating with respect to $t$ yields

$$
    \frac{d}{dt}
    F(\boldsymbol{\theta}^{(v)}(t))
    =
    \sigma(S\cup\{v\})-\sigma(S).
$$

Since
\(\boldsymbol{\theta}^{(v)}(0)=\mathbf 1_S,\)
the derivative on the left at $t=0$ is precisely the partial derivative
of $F$ with respect to $\theta_v$ evaluated at $\mathbf 1_S$.
Therefore

$$
    \frac{\partial F}{\partial \theta_v}(\mathbf 1_S)
    =
    \sigma(S\cup\{v\})-\sigma(S).
$$

\end{proof}

Consequently, for every $v\notin S$,
\(\left[\nabla F(\mathbf 1_S)\right]_v\)
is exactly the greedy marginal gain of adding $v$ to $S$.
Thus the gradient
\(\nabla F(\mathbf 1_S)\)
contains all candidate marginal gains simultaneously.

\subsection{Reverse-Mode Differentiation}

The bounded-treewidth inference procedure described above can be viewed
as an acyclic arithmetic computation whose inputs are the artificial
source-edge probabilities
\(
    \boldsymbol{\theta}
    =
    (\theta_v)_{v\in V}
\)
and whose scalar output is
\(F(\boldsymbol{\theta}).\)

Once the graph, tree decomposition, edge ownership, and relation-state
transitions have been fixed, all remaining numerical computations
consist of additions and multiplications of probability masses.
For example, combining two relation states with masses $a$ and $b$
produces a contribution
\(ab\)
to the relation state determined by their union and transitive closure.
The transitive-closure computation itself is discrete: it determines
\emph{which} output relation receives the contribution, but introduces
no additional dependence on $\boldsymbol{\theta}$.

Consequently, after the relation transitions have been compiled, the
entire inference algorithm is an arithmetic directed acyclic graph.
Let $L$ denote the number of arithmetic operations in this computation.
A forward evaluation computes
\(F(\boldsymbol{\theta})\)
in
\(O(L)\)
time.

This interpretation also gives a direct implementation-level view of
the differentiation. Suppose, for example, that the same dynamic
program is implemented using an automatic-differentiation framework
such as PyTorch \cite{pytorch}, with the vector
\(\boldsymbol{\theta}\)
declared to require gradients. The forward execution constructs the
same computation graph used to evaluate the exact influence. Calling
the reverse-mode operation on the scalar output, conceptually
\(F.\mathrm{backward()},\)
then propagates derivatives backward through all additions and
multiplications and returns

$$
    \theta_v.\mathrm{grad}
    =
    \frac{\partial F}{\partial\theta_v}
$$

for every vertex $v$ simultaneously.

This is precisely reverse-mode automatic differentiation applied to
the arithmetic circuit defined by the tree-decomposition dynamic
program. Since the circuit has one scalar output, reverse mode visits
each arithmetic operation only a constant number of times. Therefore
all $n$ partial derivatives

$$
    \frac{\partial F}{\partial\theta_1},
    \ldots,
    \frac{\partial F}{\partial\theta_n}
$$

are obtained in
\( O(L)\)
additional time, rather than requiring $n$ separate evaluations of the
inference algorithm.

One implementation detail is important at Boolean seed assignments.
For an artificial edge $\rho\to v$, a probability mass $p$ is split
into \(p(1-\theta_v)\) and \(p\theta_v\)
for the absent and present cases, respectively. Both arithmetic
branches must remain in the computation even when
$\theta_v\in{0,1}$. Numerically pruning the zero-mass branch would
remove the corresponding derivative path and could therefore destroy
the marginal-gain information. Retaining both branches preserves the
correct derivative at \(\boldsymbol{\theta}=\mathbf 1_S.\)

\begin{theorem}[All exact marginal gains]
\label{thm:all-gains}
Given a width-$w$ tree decomposition of the original graph and a
current seed set $S$, the exact marginal gains

$$
    \Delta(v\mid S)
    =
    \sigma(S\cup\{v\})-\sigma(S)
$$

for all $v\in V\setminus S$ can be computed simultaneously in

$$
    O\!\left(
        n\,2^{O(w^2)}\operatorname{poly}(w)
    \right)
$$

time.
\end{theorem}

\begin{proof}
Augment the graph with the artificial source $\rho$ and the variable
source edges \(\rho\to v,\) \( v\in V,\)
with probabilities $\theta_v$.
Adding $\rho$ to every bag increases the decomposition width from $w$
to at most $w+1$, and therefore does not change the asymptotic
width dependence:

$$
    2^{O((w+1)^2)}\operatorname{poly}(w+1)
    =
    2^{O(w^2)}\operatorname{poly}(w).
$$

By Theorem~\ref{thm:exact-ic}, a forward evaluation of
\(F(\boldsymbol{\theta})\)
on this augmented decomposition requires
$$
    O\!\left(
        n\,2^{O(w^2)}\operatorname{poly}(w)
    \right)
$$
arithmetic operations. Let $L$ denote the size of the corresponding
arithmetic circuit, so that
$$
    L
    =
    O\!\left(
        n\,2^{O(w^2)}\operatorname{poly}(w)
    \right).
$$
Evaluate this circuit at \(\boldsymbol{\theta}=\mathbf 1_S.\)
Since $F$ is scalar, reverse-mode differentiation computes the complete
gradient
\( \nabla F(\mathbf 1_S)\)
by one backward traversal of the same circuit, requiring only
$O(L)$ additional arithmetic operations.

By Theorem~\ref{thm:gradient-marginals}, for every $v\notin S$,

$$
    \frac{\partial F}{\partial\theta_v}(\mathbf 1_S)
    =
    \sigma(S\cup\{v\})-\sigma(S)
    =
    \Delta(v\mid S).
$$

Hence the single forward evaluation followed by a single reverse pass
produces all exact candidate marginal gains simultaneously, with total
running time

$$
    O\!\left(
        n\,2^{O(w^2)}\operatorname{poly}(w)
    \right).
$$

\end{proof}

Thus computing all $n-|S|$ exact marginal gains has the same
asymptotic complexity as one exact influence evaluation, up to the
constant-factor overhead of the reverse pass. In particular, the
algorithm does not perform a separate influence computation for each
candidate vertex.


\section{Exact Greedy Influence Maximization}
\label{sec:greedy}

The previous result immediately yields an exact implementation of the
classical greedy algorithm.

\begin{algorithm}[t]
\caption{Exact Greedy Influence Maximization on Bounded-Treewidth Graphs}
\label{alg:exact-greedy}
\begin{algorithmic}[1]
\Require Graph $G=(V,E)$, edge probabilities ${p_e}_{e\in E}$,
tree decomposition $\mathcal T$, seed budget $K$
\Ensure Seed set $S$

\State Introduce the artificial source $\rho$ and add $\rho$ to every bag
\State For each $v\in V$, add the variable source edge \(\rho\to v\) with influence probability \(\theta_v\)

\State Assign every original and artificial edge to one owning bag
\State  Compile the bag and separator reachability states, together with the combination, transitive-closure, and restriction maps described in Section~\ref{subsec:prob-messages} and Lemma~\ref{lem:message-correctness}.

\State $S\gets\emptyset$

\For{$i=1,\ldots,K$}
\State Set
\(        \boldsymbol{\theta}\gets\mathbf 1_S
   \)

\State Compute the local relation distributions for the current
$\boldsymbol{\theta}$

\State Perform the upward and downward sweeps to compute all exact
directed messages $M_{x\to y}$

\State Combine the incoming messages at each bag and evaluate
\[
    F(\boldsymbol{\theta})
    =
    \sum_{v\in V}
    \Prob_{\boldsymbol{\theta}}(\rho\rightsquigarrow v)
\]

\State Reverse the arithmetic computation to obtain
\[
    g_v
    \gets
    \frac{\partial F}{\partial\theta_v}
    (\mathbf 1_S),
    \qquad v\in V
\]

\State Select
\(
    v^\star
    \gets
    \arg\max_{v\notin S} g_v
\)

\State $S\gets S\cup\{v^\star\}$

\EndFor

\State \Return $S$
\end{algorithmic}
\end{algorithm}

\begin{theorem}[Running time of exact greedy]
\label{thm:greedy-runtime}
Algorithm \ref{alg:exact-greedy} computes exactly the same seed set as
classical greedy influence maximization using exact marginal gains, in
\[
    O\!\left(
        K n\,2^{O(w^2)}\operatorname{poly}(w)
    \right)
\]
time.
\end{theorem}

For fixed $K$ and $w$, the running time is
\( O(n).\)

Because IC influence is monotone and submodular, the returned solution
satisfies
\[
    \sigma(S)
    \ge
    \left(1-\frac1e\right)
    \OPT_K.
\]

\subsection{Why Not Solve Influence Maximization Exactly?}

Algorithm~\ref{alg:exact-greedy} computes the greedy solution exactly.
It does not solve globally optimal influence maximization exactly.

This distinction is unavoidable in general even on extremely simple
graph topologies.

Existing results show that IC influence maximization is NP-hard on
in-arborescences \cite{in-arborescence-IC}.
The underlying undirected graph of every arborescence is a tree, and
therefore has treewidth $1$. Furthermore, the graph in the construction of \cite{in-arborescence-IC} has its pathwidth bounded by \(2\).

\begin{corollary}
Globally optimal IC influence maximization is NP-hard even on graphs
of pathwidth at most $2$.
\end{corollary}

\begin{proof}
Lu et al.~\cite{in-arborescence-IC} prove NP-hardness of IC influence
maximization by a reduction from Subset Sum to an in-arborescence.
Their construction consists of paths
\[
v_{i,1},v_{i,2},\ldots,v_{i,L_1+L_2},
\qquad i=1,\ldots,n,
\]
whose final vertices are adjacent to a common vertex $g_1$, followed by
the path
\[
g_1,g_2,\ldots,g_{L_3}.
\]

Let $L=L_1+L_2$. A path decomposition of the constructed graph is
obtained by taking, consecutively for each $i=1,\ldots,n$, the bags
\[
\{v_{i,j},v_{i,j+1},g_1\},
\qquad
j=1,\ldots,L-1,
\]
and then appending
\[
\{g_1,g_2\},
\{g_2,g_3\},
\ldots,
\{g_{L_3-1},g_{L_3}\}.
\]

Every graph edge is contained in some bag, and the bags containing any
fixed vertex form a contiguous interval. The largest bag has size
three, so the constructed graph has pathwidth at most $2$.

Since the reduction of Lu et al. is NP-hard and all graphs produced by
the reduction have pathwidth at most $2$, globally optimal IC influence
maximization is NP-hard even on graphs of pathwidth at most $2$.
\end{proof}

Thus bounded pathwidth/treewidth makes exact probabilistic inference tractable,
but does not eliminate the combinatorial hardness of globally optimal
seed selection. This contrast is summarized in Table~\ref{tab:complexity-summary}.

Interestingly, the IM problem under the Linear Threshold diffusion model is solvable exactly in polynomial time on
in-arborescences \cite{in-arborescence-LT}. Thus, determining whether exact IM under the LT model is tractable on bounded-pathwidth or bounded-treewidth graphs remains an interesting direction for future research.

\begin{table}[t]
\centering
\small
\begin{tabular}{l l}
\toprule
\textbf{Task} & \textbf{Complexity} \\
\midrule
Exact influence evaluation
& $O(nf(w))$ \\

All exact marginal gains
& $O(nf(w))$ \\

Exact greedy, budget $K$
& $O(Knf(w))$ \\

Globally optimal influence maximization
& NP-hard for $\tw(G)=1$ \\
& NP-hard for $\pw(G)\le 2$ \\
\bottomrule
\end{tabular}
\caption{Complexity of exact IC inference and influence maximization on
graphs of bounded width. Here \(f(w)=2^{O(w^2)}\operatorname{poly}(w).\)}
\label{tab:complexity-summary}
\end{table}


\section{Probabilistic Seed Activation}
\label{sec:uncertain-seeds}

In many applications, selecting a node as a target does not guarantee
that the node becomes active. It is more realistic to consider interventions as having a probability of success.
Let \(q_v\in[0,1]\)
be the probability that targeted node \(v\) becomes active when selected for the seed set.

If $\theta_v$ represents whether $v$ is selected as a target (i.e. \(\theta_v\in \{0,1\}\)), the
artificial source edge becomes \(\rho\xrightarrow{q_v\theta_v}v.\)

The exact bounded-treewidth circuit is unchanged structurally.
Only its numerical edge weights change.

Consequently, the asymptotic running time of the treewidth algorithm is
independent of the numerical values of \(p_e\) and \(q_v\).

This contrasts with sampling-based methods, whose running times may
depend strongly on cascade sizes, estimator variance, RR-set sizes, and
the estimated optimal spread.


\section{Experiments}
\label{sec:experiments}

We conducted a series of experiments to empirically evaluate the
behavior of the proposed exact bounded-treewidth influence-maximization
algorithm. Our aim is not to provide an exhaustive empirical study,
but rather to demonstrate the practical feasibility of the method and
to investigate its scaling behavior under different diffusion regimes.

In particular, we study how the running time depends on graph size and
treewidth, how the methods respond to different propagation and seed
activation probabilities, and how the proposed exact greedy algorithm
compares with representative scalable influence-maximization
baselines on bounded-treewidth instances.

\paragraph{Baselines.}
We compare the following methods:
\begin{itemize}
\item \textbf{TW exact greedy}, the proposed method;
\item \textbf{IMM}~\cite{IMM};
\item a \textbf{directed NewGreedy (NG)/Cohen-style all-gains Monte
Carlo} method~\cite{chen2009efficient,cohen1997size}; and
\item \textbf{CELF with Monte Carlo influence estimation}
~\cite{CELF}.
\end{itemize}

For a fair comparison of solution quality, the seed set returned by
each method is evaluated afterwards using the same exact
bounded-treewidth influence evaluator. This post-hoc evaluation is not
included in the optimization running time of any method.

CELF with Monte Carlo estimation was included in preliminary
experiments, but its running time became noncompetitive already at
substantially smaller graph sizes than those reached by the remaining
methods. We therefore omit it from the large-scale results presented
below.

\subsection{Additional Benchmark Details}

All experiments were run on a Dell Precision 7710 workstation with an
Intel Core i7-6820HQ processor (8 logical CPUs), 16\,GB of RAM, and a
256.1\,GB disk, running Ubuntu 24.04.4 LTS.

For the experiments reported below, we considered treewidths
$w\in\{1,2\}$ and graph sizes ranging from $n=300$ to $n=100000$,
subject to computational feasibility. The reported diffusion and
seed-success regimes were
\[
(a,q)\in
\{(0.1,1),(0.5,1),(1,1),(1,0.1),(1,0.05)\},
\]
with seed budget $K=10$.
Not every combination of treewidth, stochastic regime, and graph size
was evaluated up to the same maximum $n$; Table~\ref{tab:tw_benchmark_headline}
therefore reports the largest completed instance available for each
setting. All recorded method runs completed with successful status, and
no method-level timeout was observed. We therefore do not attribute the
termination of the larger benchmark to any particular method; rather,
it appears to have resulted from overall hardware-resource limitations
of the experimental machine.  The underlying graphs were random $k$-trees
of exact treewidth $k=w$, converted to bidirected IC graphs.  Each
directed edge probability was sampled independently as
$p_e\sim U(0,a)$.

IMM was run with $\varepsilon=0.1$ and $\ell=1$ and without an
explicit cap on the number of RR sets.  The directed
NewGreedy/Cohen-style baseline used 500 live-edge snapshots per greedy
round and five Cohen exponential-min sketch repetitions per snapshot.
Thus, a completed $K=10$ NewGreedy run used 5000 live-edge snapshots.

Each method was given an independent wall-clock limit of 3600 seconds.
The treewidth circuit compilation had a separate timeout of 1800
seconds.  A timeout of one method did not prevent the remaining methods
from being evaluated.  Each parameter setting was run once.

For every returned seed set, solution quality was evaluated afterwards
using the same exact bounded-treewidth influence evaluator.  This
post-hoc evaluation time was excluded from the reported optimization
runtime.  For the proposed method, the reported total runtime includes
both the one-time treewidth-circuit compilation and exact greedy seed
selection, whereas the reported IMM and NewGreedy runtimes contain
seed selection only.

For the large-scale benchmark, we attempted graph sizes \(n\in\{30000,100000,300000,1000000\}\) for treewidths \(w\in\{1,2\}\). Not all requested configurations completed on the available hardware. The missing configurations were not caused by the configured per-method wall-clock timeout: no corresponding method-level timeout was recorded before termination of the benchmark process. We therefore report only configurations for which complete results were successfully produced. The termination appears to have resulted from overall computational-resource limitations of the experimental machine rather than a recorded timeout of any individual method.

Because termination occurred outside the method-level status reporting, we do not attribute the missing configurations to any particular method.

\subsection{Results}

\paragraph{Runtime Scaling.}
The experiments were consistent with the linear dependence on graph
size predicted for Algorithm~\ref{alg:exact-greedy} at fixed treewidth
and seed budget. In the more demanding diffusion regimes, our method outperformed the sampling-based baselines, by several orders of magnitude in the most extreme cases. Figure~\ref{fig:runtime-scaling} illustrates the
runtime scaling in a particularly challenging stochastic regime, with
strong diffusion ($a=1$) and a low seed activation probability
($q=0.05$).

 \begin{figure}[t]
 \centering
 \includegraphics[width=.8\linewidth]{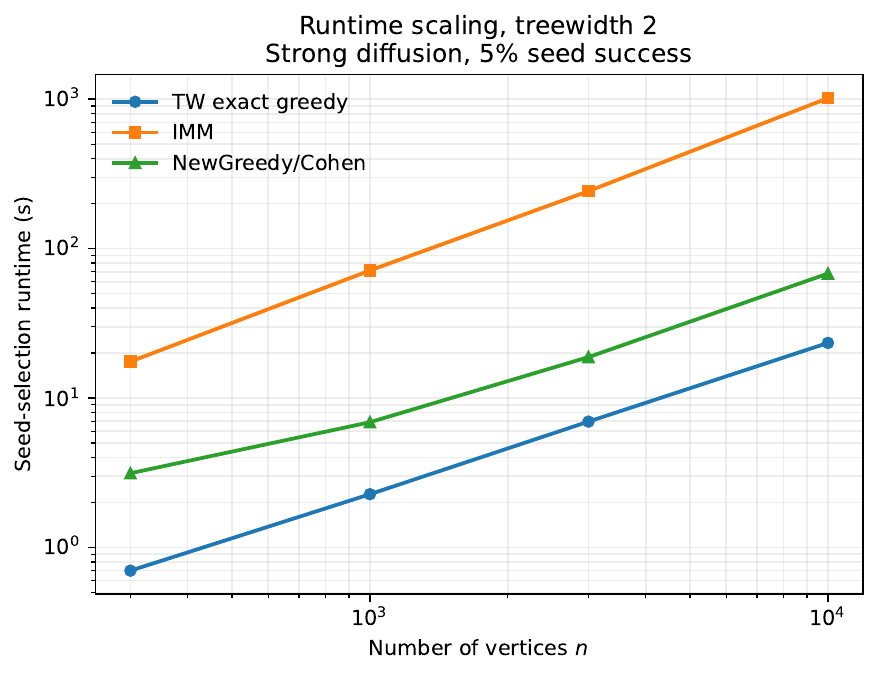}
 \caption{
 Runtime as a function of graph size for treewidth $2$ under strong
 diffusion and probabilistic seed activation.
 }
 \label{fig:runtime-scaling}
 \end{figure}

\paragraph{Robustness to Stochastic Parameters.}
Figure~\ref{fig:probability-robustness} shows that the runtime of
Algorithm~\ref{alg:exact-greedy} is largely insensitive to the
particular values of the propagation parameter $a$ and seed activation
probability $q$. This behavior is expected: these parameters change the
numerical probability masses propagated by the algorithm, but not the
underlying set of reachability states or the structure of the compiled
computation.

As further illustrated in Table~\ref{tab:tw_benchmark_headline}, the
dominant factors determining the runtime of our method are instead the
size and treewidth of the graph. In contrast, the runtimes of IMM and
the NewGreedy-style Monte Carlo method vary substantially across
stochastic regimes, since their computational requirements depend on
the sampling behavior induced by the propagation and seed activation
probabilities. This insensitivity to the numerical diffusion parameters
is therefore an important practical advantage of the exact
bounded-treewidth approach.

 \begin{figure}[t]
 \centering
 \includegraphics[width=.8\linewidth]{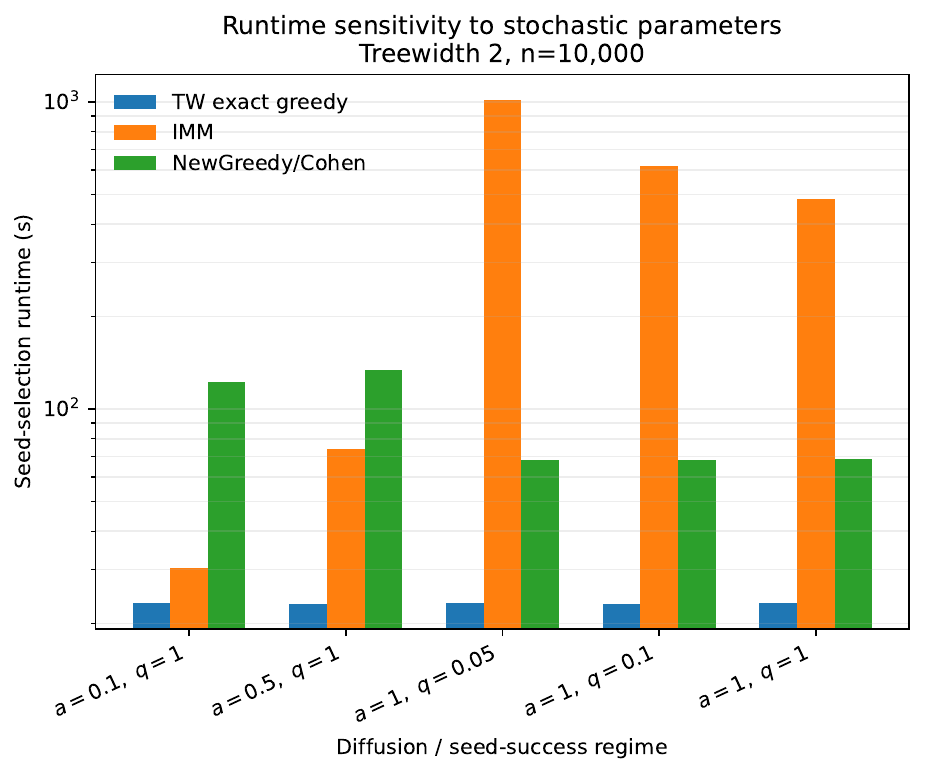}
 \caption{
 Runtime across propagation and seed-activation regimes.
 }
 \label{fig:probability-robustness}
 \end{figure}

\paragraph{Solution Quality.}
Table~\ref{tab:tw_benchmark_headline} shows that all evaluated methods
generally produced seed sets with very similar exact influence spread.
Algorithm~\ref{alg:exact-greedy} most consistently achieved the highest
observed spread. Unlike the sampling-based baselines, it computes every
greedy marginal gain exactly and therefore reproduces the exact
classical greedy trajectory without sampling error.

This does not imply that exact greedy always returns the globally
optimal seed set: an approximate method may occasionally select a
different sequence of vertices and obtain a seed set with slightly
higher final spread. Nevertheless, the results indicate that the
sampling-based methods require substantially greater computational
effort to attain solution quality close to that of exact greedy on the
bounded-treewidth instances considered here.

\begin{table}
\caption{Runtime and solution quality at the largest available graph size for each experimental regime. The runtime is measured in seconds (s), and the solution quality is measured in percentage of the best found solution of all baselines (\%).}
\small
\label{tab:tw_benchmark_headline}
\begin{tabular}{rrrrrrrrrr}
\toprule
w & n & a & q & TW (s) & IMM (s) & NG (s) & TW (\%) & IMM (\%) & NG (\%) \\
\midrule
1 & 100000 & 0.10 & 1.00 & 154.54 & 659.36 & 1502.84 & 100.00 & 99.98 & 94.28 \\
1 & 10000 & 0.50 & 1.00 & 14.24 & 17.08 & 130.01 & 100.00 & 100.00 & 97.34 \\
1 & 100000 & 1.00 & 0.05 & 153.65 & 608.95 & 1643.90 & 100.00 & 99.29 & 100.00 \\
1 & 10000 & 1.00 & 0.10 & 14.07 & 44.63 & 118.47 & 100.00 & 99.86 & 99.76 \\
1 & 100000 & 1.00 & 1.00 & 153.17 & 91.13 & 1596.28 & 100.00 & 99.66 & 97.13 \\
2 & 10000 & 0.10 & 1.00 & 23.26 & 30.28 & 122.03 & 100.00 & 99.81 & 99.52 \\
2 & 10000 & 0.50 & 1.00 & 23.22 & 74.24 & 134.45 & 98.99 & 100.00 & 96.38 \\
2 & 10000 & 1.00 & 0.05 & 23.38 & 1016.79 & 68.08 & 100.00 & 99.08 & 99.83 \\
2 & 10000 & 1.00 & 0.10 & 23.24 & 617.10 & 68.07 & 100.00 & 99.09 & 99.46 \\
2 & 10000 & 1.00 & 1.00 & 23.34 & 483.93 & 68.85 & 100.00 & 99.90 & 98.59 \\
\bottomrule
\end{tabular}
\end{table}


\section{Discussion and Conclusion}
\label{sec:discussion-and-conclusion}

We studied Independent Cascade influence maximization on graphs of
bounded treewidth. Given a width-$w$ tree decomposition, exact IC
influence evaluation can be performed in

$$
    O\!\left(
        n\,2^{O(w^2)}\operatorname{poly}(w)
    \right)
$$

time by propagating distributions over separator reachability
relations.

By introducing variable artificial source edges, we further showed
that every greedy marginal gain can be expressed as a partial
derivative of a single multilinear influence function. Reverse-mode
differentiation through the exact inference computation therefore
produces all node marginal gains simultaneously, within the same
asymptotic complexity as one exact influence evaluation.

This yields an exact implementation of classical greedy influence
maximization with running time

$$
    O\!\left(
        K n\,2^{O(w^2)}\operatorname{poly}(w)
    \right).
$$

Consequently, for fixed treewidth and fixed seed budget, exact greedy
influence maximization is linear in the size of the graph. Importantly,
``exact'' here refers to the greedy marginal gains and the resulting
classical greedy trajectory, rather than to globally optimal influence
maximization.

This tractability stands in contrast to the hardness of the underlying
combinatorial optimization problem: globally optimal influence
maximization remains NP-hard already on graphs of treewidth one and
pathwidth two. The results therefore highlight a sharp distinction
between exact probabilistic inference and exact combinatorial
optimization on structurally simple graphs.

A useful feature of the proposed approach is that it shifts the main
computational difficulty from stochastic sampling to structural
complexity. Once the graph topology, tree decomposition, and
reachability-state transitions have been fixed, changing the numerical
edge or seed-activation probabilities changes only the propagated
probability masses, not the combinatorial state space. This is reflected
in our experiments, where the runtime of the exact treewidth method was
largely insensitive to the stochastic parameters, while the
sampling-based baselines varied substantially across diffusion regimes.
The experiments were also consistent with the predicted linear scaling
in graph size for fixed width.

The main limitation is the dependence on treewidth. Although in
practice only reachable relation states need to be materialized, the
number of distinguishable separator behaviors can be
\(
    2^{\Theta(w^2)}
\)
in the worst case within the context-independent compositional
separator-summary framework considered here. The method can therefore
become impractical rapidly as the treewidth increases. Its natural
domain is consequently graphs of very small treewidth, particularly in
regimes where exactness is valuable or where sampling-based methods
become expensive.

Our present experiments use synthetic graphs with controlled
treewidth, allowing the structural dependence of the algorithm to be
isolated cleanly. An important direction for future work is to identify
natural networks, applications, or decomposition-based subproblems with
sufficiently small effective width for exact inference to remain
practical. More generally, reducing the practical dependence on
treewidth, for example through improved state compression or hybrid
exact--approximate methods, remains an interesting direction.


\section{Future Work}
\label{sec:future}

Several directions appear promising.

\paragraph{Approximate separator inference.}
The exact algorithm propagates complete distributions over separator
reachability relations, which becomes expensive as the treewidth
increases. A natural approximation is to restrict the amount of
information propagated between bags, for example by retaining only a
selected subset of relation states or by truncating propagation beyond
a bounded number of decomposition-tree hops. Such methods could
interpolate between exact bounded-treewidth inference and inexpensive
local approximations, while potentially retaining much of the structural
information captured by the exact method.

\paragraph{Other diffusion models.}
The live-graph interpretation suggests that the framework may extend
beyond the Independent Cascade model. The Linear Threshold model is
particularly interesting. Under its triggering-set representation, the
diffusion process can again be expressed through random reachability,
but the possible incoming triggering edges of a vertex are locally
correlated rather than independent. Extending the separator recursion
to accommodate these local dependencies therefore appears feasible but
requires additional state or factor information.

The Linear Threshold model is also interesting from an optimization
perspective. Influence maximization on in-arborescences exhibits a
sharp complexity distinction between the IC and LT models: the problem
is tractable under LT while remaining hard under IC. This motivates the
question of whether globally optimal LT influence maximization admits
polynomial-time algorithms on graphs of bounded treewidth or bounded
pathwidth.

\paragraph{Continuous optimization.}
The variable-source construction does more than provide the discrete
greedy marginal gains. The function
\(
    F(\boldsymbol{\theta})
\)
is the multilinear extension associated with independent probabilistic
seed selection, and the proposed inference procedure provides exact
values and gradients of this function. This suggests using the same
machinery as an exact gradient oracle inside continuous-greedy or other
continuous optimization procedures. Investigating whether this leads
to useful algorithms in practice, and whether the structural setting
permits stronger guarantees or more efficient implementations, is a
natural direction for further work.

\paragraph{Faster maximum-marginal selection.}
The current algorithm computes the complete gradient and therefore all
$n$ marginal gains in every greedy iteration. While this is optimal if
all gains are explicitly required, greedy influence maximization needs
only the largest one. It is therefore natural to ask whether
\(
    \arg\max_{v\notin S}\Delta(v\mid S)
\)
can be identified in
\(
    o(nf(w))
\)
time without materializing the complete gradient. This question is
particularly interesting in combination with lazy greedy techniques
such as CELF, where previously computed marginal gains provide upper
bounds after the seed set grows. More generally, it may be possible to
exploit the fact that consecutive greedy iterations differ by only one
new seed and update the compiled inference computation incrementally.

\paragraph{Compact separator representations.}
The lower bound of
\(
    2^{\Omega(w^2)}
\)
applies to exact context-independent compositional summaries that must
preserve arbitrary external reachability contexts. Restricted graph
families may admit substantially smaller state spaces, and practical
instances may realize only a small fraction of all theoretically
possible separator relations. Identifying structural conditions under
which stronger compression is possible is therefore of both theoretical
and practical interest. A broader open question is whether exact
bounded-treewidth IC inference can avoid explicit separator-relation
states altogether by using a fundamentally different, possibly more
global, representation.


\bibliographystyle{plain}
\bibliography{references}

\end{document}